%% file: ITM_vs_RTM_arxiv.tex
\documentclass{llncs}
\usepackage{hyperref}
\hypersetup{bookmarks=true,bookmarksnumbered=true,bookmarksopen}
\usepackage{amsmath}
\usepackage[usenames]{color}\usepackage{amssymb}
\usepackage{procalg}
\usepackage{amscd}
\usepackage{txfonts}
\usepackage{amsfonts,semantic,colortbl,mathrsfs,stmaryrd,mathtools}
\usepackage{epsfig,graphicx,subfigure}
\usepackage{enumerate,enumitem}
\usepackage{multirow,multicol}
\usepackage{tabularx}
\usepackage{gastex}
\newcommand{\notinfull}[1]{}
\newcommand{\infull}[1]{#1}
\input{macro}
\title{On the Executability of Interactive Computation}
\author{Bas Luttik \and Fei Yang}
\institute
{Eindhoven University of Technology, The Netherlands
}

\begin{document}
\maketitle
\begin{abstract}
The model of \emph{interactive Turing machines} (ITMs) has been proposed to
characterise which stream translations are \emph{interactively computable}; the model of  \emph{reactive Turing machines}
(RTMs) has been proposed to characterise which behaviours are \emph{reactively
  executable}. In this article we provide a comparison of
the two models. We show, on the one hand, that the behaviour
exhibited by ITMs is reactively executable, and, on the other hand,
that the stream translations naturally associated with RTMs are
interactively computable. We conclude from these results that the
theory of reactive executability subsumes the theory of interactive
computability. Inspired by the existing model of ITMs with advice, which
provides a model of evolving computation, we also consider RTMs with
advice and we establish that a facility of advice considerably upgrades the
behavioural expressiveness of RTMs: every countable transition system
can be simulated by some RTM with advice up to a fine notion of
behavioural equivalence.
\end{abstract}
\input{Introduction}
\input{Preliminaries}
\input{ITM}

\input{ExeITM}
\input{Advice}
\input{Conclusion}
\bibliographystyle{splncs03}
\bibliography{ITMvsRTM}
\notinfull{
\input{Appendix}
}
\end{document}

%% file: macro.tex
\usepackage{color}
\makeatletter
\@ifundefined{ifAP}{%

\newcommand{\todo}[1]{\underline{\sf todo}: {\color{red}{#1}}}
}{
\ifAP
\else
\fi
\newcommand{\todo}[1]{}
}
\makeatother

\newenvironment{tarray}[2][c]{
  \settowidth{\dimen1}{${} = {}$}%
  \setlength{\arraycolsep}{0.125\dimen1}%
  \hspace{-0.125\dimen1}\array[#1]{#2}\relax
}
{
   \endarray\hspace{-0.125\dimen1}
}
\def\tbb{\begin{tarray}{llll}}
\def\tbbt{\begin{tarray}[t]{llll}}
\def\tee{\end{tarray}}

\newcommand{\myparagraph}[1]{\noindent\textbf{#1}}

\newcommand{\A}{\mathcal{A}}
\newcommand{\Atau}{\mathcal{A}_{\tau}}

\newcommand{\R}{\mathcal{\mathop{R}}}
\newcommand{\T}{\mathcal{T}}
\newcommand{\M}{\mathcal{M}}
\newcommand{\D}{\mathcal{D}}

\newcommand{\Dbox}{\mathcal{D}_{\Box}}
\newcommand{\Sta}{\mathcal{S}}

\newcommand{\C}{\mathcal{C}}

\newcommand{\Out}{\mathit{O}}
\newcommand{\In}{\mathit{I}}
\newcommand{\Exe}{\mathit{E}}

\newcommand{\I}{\mathcal{I}}

 \newcommand{\tphdL}[1]{{#1 \!}^{\scriptscriptstyle <}}
 \newcommand{\tphdR}[1]{\prescript{\scriptscriptstyle >}{}{\! #1}}

\newcommand{\bbisimd}{\ensuremath{\mathrel{\bbisim^{\Delta}}}}
\newcommand{\witequ}{\ensuremath{\mathrel{=_{\text{wt}}^{\infty}}}}

\newcommand{\qin}{q_{\mathalpha{in}}}

\newcommand{\witr}{\ensuremath{\mathit{Tr_w^{\infty}}}}
\newcommand{\encode}[1]{\ensuremath{\ulcorner {#1}\urcorner}}
\newcommand{\fullversion}[1]{} 
\newcommand{\delete}[1]{}

%% file: Introduction.tex
\section{Introduction}\label{sec:intro}

According to the Church-Turing thesis, the classical Turing machine model adequately formalises which functions from natural numbers to natural numbers are effectively computable. There is, however, a considerable semantic gap between computing the result of a function applied to a natural number and the way computing systems operate nowadays. Modern computing systems are reactive, they are in continuous interaction with their environment, and their operation is not supposed to terminate. Quite a number of extended models of computation have been proposed in recent decades to study the combination of computation and interaction (see, e.g., the collection in \cite{GSW2006}). In this paper we compare \emph{interactive Turing machines} and \emph{reactive Turing machines}.

Van Leeuwen and Wiedermann have developed a theory of interactive computation from the stance that an interactive computation can be viewed as a never-ending exchange of symbols between a component and its unpredictable interactive environment \cite{vLW2000}. Semantically, this amounts to studying the recognition, generation and translation of infinite streams of symbols. In \cite{vLW2001}, the notion of interactive Turing machine (ITM) is put forward as a tool to formally characterise which stream translations are interactively computable. The notion is subsequently extended with an (non-computable) advice mechanism in order to obtain a non-uniform machine model. Van Leeuwen and Wiedermann argue that the resulting model of \emph{interactive Turing machines with advice} is as powerful as their model of evolving finite automata, and they conclude from this, on intuitive grounds, that ITMs with advice are adequate to model evolving system such as the Internet \cite{vLW2008}.

The model of interactive Turing machines focusses on capturing the computational content of sequential interactive behaviour. The included mechanism of interaction is therefore limited to achieving this goal, and does not easily generalise to more than one distributed component, nor does it allow for more fine-grained considerations of the behaviour of reactive systems. The behavioural theory of reactive systems, on the other hand, has focussed on aspects of modelling, specification and verification (see, e.g., \cite{AILS2007}).

To integrate computability theory and the behavioural theory of reactive systems, the notion of reactive Turing machine (RTM) has been proposed in \cite{BCLT2009,BLT2013}. It extends Turing machines with concurrency-style interaction. Semantically, the operational behaviour of an RTM is given by a transition system. From this transition system one may extract a set of computations, or stream translations,  but a more refined analysis is also possible. In fact, to study the effect of interaction of multiple components many refined notions of behavioural equivalence have been developed in the concurrency theory literature \cite{Glabbeek1993}. The notion of RTM gives rise to a general theory of \emph{executability}: a transition system is executable (usually up to some preferred notion of behavioural equivalence) if there exists an RTM that has the transition system as its semantics. (We refer to \cite{BLT2013} for more a elaborate motivation of the notion of RTM.)

The aim of this paper is to make a connection between the theory of interactive computabililty and the theory of reactive systems, providing a comparison of the models of ITMs and RTMs in both their semantic domains. We shall first, in Section~\ref{sec:pre}, recapitulate both models. Then, in Section~\ref{sec:itm} we present a transition-system semantics for ITMs; the transition system associated with an ITM is executable up to a fine notion of behavioural equivalence. In Section~\ref{sec:exeomega} we shall identify a subclass of RTMs that can be considered suitable for stream translation, and prove that the stream translation associated with an RTM in this subclass is interactively computable. In Section~\ref{sec:advice} we consider an extension of RTMs with an advice mechanism adapted from the advice mechanism considered for ITMs. RTMs with advice can execute every countable transition system, at the cost of introducing divergence in the computation. The paper ends with a conclusion in Section~\ref{sec:conclusion}.

\notinfull{
For the reviewers' convenience, proofs of the results and some standard formal definitions (also to be found in the cited literature) are included in an appendix. A full version of this submission with proofs of all results is available as~\cite{LY16}.} 

%% file: Preliminaries.tex
\section{Preliminaries}\label{sec:pre}
\delete{
In this section, we briefly introduce the theories of interactive computability and executability, and define a subclass of transition systems with which $\omega$-translations can be associated.
}
\subsection{The Theory of Interactive Computation}

In~\cite{vLW2006}, van Leeuwen and Wiedermann present an analysis of interactive computation on the basis of a \emph{component} $C$ (thought to behave according to a deterministic program) interacting with an unpredictable \emph{environment} $E$. They discuss the consequences of a few general postulates pertaining to the behaviour and interaction of $C$ and $E$ for interactive recognition, interactive generation and interactive translation. In their analysis, the component $C$ acts as a stream transducer, transforming an infinite input stream of data symbols from  $\Sigma=\{0,1\}$ presented by $E$ at its input port into an infinite output stream of symbols from $\Sigma$ produced at its output port. Henceforth, by an \emph{$\omega$-translation} we mean a mapping $\phi: \Sigma^{\omega}\rightarrow\Sigma^{\omega}$ (with $\Sigma^{\omega}$ denoting the set of streams, i.e., infinite sequences, over $\Sigma$).

Interactive computation is a step-wise process. It is not required that the environment offers a symbol in every step, nor that the component produces a symbol in every step. For the purpose of modelling components, however, it is convenient to record that nothing is offered or produced. The symbol $\lambda$ is used to indicate the situation that no symbol is offered at the input port or produced at the output port, and we let $\Sigma_\lambda=\Sigma\cup\{\lambda\}$. It is assumed that when $E$ offers a non-$\lambda$ symbol in some step, then the component $C$ produces a non-$\lambda$ symbol at its output port within finitely many steps, and vice versa; this assumption is referred to as the \emph{interactiveness} (or \emph{finite delay}) condition in the work of van Leeuwen and Wiedermann.


In order to formally define which $\omega$-translations are interactively computable by a computational device, van Leeuwen and Wiedermann proposed the notion of \emph{interactive Turing machine} \cite{vLW2001,vLW2001a}. It extends the classical notion of Turing machine with an input port and an output port, through which it exchanges an infinite, never ending stream of data symbols with its environment. Interactive Turing machines use a two-way infinite tape as memory on which they can write symbols from some presupposed set $\Dbox$ of \emph{tape symbols}, not necessarily disjoint from $\Sigma$ and including the special $\Box$ symbol to denote an empty tape cell. Our formal definition below is adapted from~\cite{Verbaan2006} (but we leave out the distinction between internal and external states).

\begin{definition}\label{def:itm}
A \emph{(deterministic) interactive Turing machine} (ITM) with a single work tape is a triple $\I=(Q,\step{}_{\I},\qin)$, where
\begin{enumerate}
    \item $Q$ is its set of \emph{states};
    \item $\step{}_{\I}:Q\times\Dbox\times\Sigma_\lambda\rightarrow Q\times\Dbox\times\{L,R\}\times\Sigma_\lambda$ is a transition function; and
    \item $\qin\in Q$ is its \emph{initial state}.
\end{enumerate}
\end{definition}

The contents of the tape of an ITM may be represented by an element of $(\Dbox)^{*}$\fullversion{ (the set of finite sequences of symbols in $\Dbox$)}. We denote by $\check{\Dbox}=\{\check{d}\mid d\in\Dbox\}$ the set of \emph{marked} symbols; a \emph{tape instance} is a sequence $\delta\in(\Dbox\cup\check{\Dbox})^{*}$ such that $\delta$ contains exactly one element of $\check{\Dbox}$. The marker indicates the position of the tape head.
\fullversion{
Intuitively the transition $(q,\delta)\step{i/o}_{\I}(q',\delta')$ means that whenever the ITM is in state $q$, its tape head reads the symbol $d$, and input symbol $i$ is offered at its input port, then it replaces the symbol $d$ by the symbol $e$ on its tape, moves the tape head one position in the direction of $M$, produces the output symbol $o$ at its output port and then continues in state $q'$.}

A \emph{computation} of an ITM $\I=(Q,\step{}_{\I},\qin)$ is an infinite sequence of transitions
\begin{equation} \label{eq:itmcomputation}
  (\qin,\check{\Box})=(q_0,\delta_0)\step{i_0/o_0}_{\I}(q_1,\delta_1)\step{i_1/o_1}_{\I}\cdots (q_k,\delta_k)\step{i_k/o_k}_{\I}\cdots
\enskip.
\end{equation}
The \emph{input stream} \delete{\todo{It may not be a stream, but a finite sequence; we may want to adapt terminology.}} associated with the computation in \eqref{eq:itmcomputation} is obtained from $i_0,i_1,\dots$ by omitting all occurrences of $\lambda$, and the \emph{output stream} associated with the computation in \eqref{eq:itmcomputation} is obtained from $o_0,o_1,\dots$ by omitting all occurrences of $\lambda$. A pair $(\vec{x},\vec{y})\in\Sigma^{\omega}\times\Sigma^{\omega}$ is an \emph{interaction pair} associated with $\I$ if there exists a computation of $\I$ with $\vec{x}$ as input stream and $\vec{y}$ as output stream. The set of all interaction pairs associated with an ITM $\I$ is called its \emph{interactive behaviour}. (In Section~\ref{sec:itm} we shall present a more refined view on its behaviour when we associate with every ITM a transition system.) The computation in \eqref{eq:itmcomputation} is \emph{interactive} if, for all $k\in\mathbb{N}$, if $i_k\neq\lambda$, then there exists $\ell\geq k$ such that $o_{\ell}\neq\lambda$. The computation in \eqref{eq:itmcomputation} is \emph{input-active} if $i_k\neq\lambda$ for all $k\in\mathbb{N}$.

An ITM satisfies the \emph{interactiveness} condition if all its computations are interactive. Clearly, if a deterministic ITM $\I$ satisfies the interactiveness condition, then its interactive behaviour is total, in the sense that for every $\vec{x}\in\Sigma^{\omega}$ there is at least one $\vec{y}\in\Sigma^{\omega}$ such that $(\vec{x},\vec{y})$ is an interaction pair of $\I$. By confining our attention to the input-active computations---which, in the terminology of \cite{vLW2006}, corresponds to adopting the full environmental activity postulate---, we may then associate with every such ITM an $\omega$-translation: we say that ITM $\I$ produces $\vec{y}$ on input $\vec{x}$ if $(\vec{x},\vec{y})$ is the interaction pair associated with an input-active computation of $\I$.

\begin{definition}
  An $\omega$-translation $\phi:\Sigma^{\omega}\rightarrow\Sigma^{\omega}$ is \emph{interactively computable} if there exists a deterministic ITM that satisfying the interactiveness condition that produces $\phi(\vec{x})$ on input $\vec{x}$ for all $\vec{x}\in\Sigma^{\omega}$.
\end{definition}

\infull{
Van Leeuwen and Wiedermann present in~\cite{vLW2006} a characterisation of the interactively computable $\omega$-translations by showing that they can be approximated by classically computable partial functions on finite sequences over $\Sigma$.
For finite and infinite sequences $\vec{x}$ and $\vec{y}$, we write $\vec{x}\prec\vec{y}$ if $\vec{x}$ is a finite and strict prefix of $\vec{y}$, and $\vec{x}\preceq\vec{y}$ if $\vec{x}\prec\vec{y}$ or $\vec{x}=\vec{y}$. We use the following definition of monotonic functions and limit-continuous functions.
\begin{definition}\label{def:mono-limit-cont}
\begin{enumerate}
\item A partial function $f:\Sigma^{*}\rightharpoonup\Sigma^{*}$ is \emph{monotonic} if for all $\vec{x},\vec{y}\in\Sigma^{*}$ such that $\vec{x}\prec\vec{y}$ and $f(\vec{y})$ is defined, it holds that $f(\vec{x})$ is defined as well and $f(\vec{x})\preceq f(\vec{y})$.
\item A partial function $\phi: \Sigma^{\omega}\rightarrow\Sigma^{\omega}$ is called \emph{limit-continuous} if there exists a classically computable monotonic partial function $f:\Sigma^{*}\rightarrow\Sigma^{*}$ such that
  $\phi(\lim_{k\rightarrow\infty}\vec{x}_k)=\lim_{k\rightarrow\infty}f(\vec{x}_k)$
for all strictly increasing chains $\vec{x}_1\prec \vec{x}_2\prec\cdots\prec \vec{x}_k\prec\cdots$ with $\vec{x}_k\in\Sigma^{*}$.
\end{enumerate}
\end{definition}}

In~\cite{vLW2006} a criterion of the interactively computable $\omega$-translations is presented by using limit-continuous functions\notinfull{ (a formal definition is included in the Appendix)}.

\begin{theorem}\label{thm:int-trans}
  A total $\omega$-translation is interactively computable iff it is limit-continuous.
\end{theorem}

\subsection{The Theory of Executability}
\fullversion{
\todo{Reorganise this section (possibly reusing some text to found elsewhere in the paper). Its main purpose is to briefly recap the theory of executability introduced in \cite{BLT2013} (as far as we need it in this paper). I propose the following structure for this section:
\begin{enumerate}
  \item A brief account of the motivation for the theory of executability: it tries to combine computation and concurrency-style interaction in such a way that both are treated on equal footing; thus, an integration of computability and concurrency theory is realised.
  \item Introduce transition systems as the mathematical representation of behaviour.
  \item Introduce reactive Turing machines in order to define the notion of executable transition system. Note that some of the notions (e.g., how to deal with the tape, what is a configuration, etc.) have already been introduced and explained in the preceding subsection, so we need not elaborate to much about these issues and can refer back to some extent.
  \item Explain that transition systems are generally considered up to some behavioural equivalence relation, and introduce divergence preserving branching bisimilarity as the finest useful notion.
  \item As an illustration (and to have it available for the remainder of the paper), establish the result that reactive Turing machines with 'stay-transitions' (no movement of the tape head) are not more powerful than reactive Turing machines without, up to divergence-preserving branching bisimilarity.
\end{enumerate}}
}
The theory of executability combines computation and concurrency-style interaction in such a way that both are treated on equal footing; thus, an integration of computability and concurrency theory is realised.

The transition system is the central notion in the mathematical theory of discrete-event behaviour. It is parameterised by a set $\A$ of \emph{action symbols}, denoting the observable events of a system. \delete{We shall later impose extra restrictions on $\A$, e.g., requiring that it be finite or has a particular structure, but for now we let $\A$ be just an arbitrary abstract set.} We extend $\A$ with a special symbol $\tau$, which intuitively denotes unobservable internal activity\delete{ of the system}. We shall abbreviate $\A \cup\{\tau\}$ by $\Atau$.

\begin{definition}\label{def:lts}
An \emph{$\Atau$-labelled transition system} $\T$ is a triple $(\Sta,\step{},\uparrow)$, where,
\begin{enumerate}
    \item $\Sta$ is a set of \emph{states},
    \item ${\step{}}\subseteq\Sta\times\Atau\times\Sta$ is an $\Atau$-labelled \emph{transition relation},
    \item ${\uparrow}\in\Sta$ is the initial state.
\end{enumerate}
\fullversion{If $(s,a,t)\in{\step{}}$, then we write $s\step{a} t$.}
\end{definition}%
%

Transition systems can be used to give semantics to programming languages and process calculi. The standard method is to first associate with every program or process expression a transition system (its operational semantics), and then consider programs and process expressions modulo one of the many behavioural equivalences on transition systems that have been studied in the literature. In this paper, we shall use the notion of (divergence-preserving) branching bisimilarity \cite{Glabbeek1996,Glabbeek2009}, which is the finest behavioural equivalence in van Glabbeek's linear time - branching time spectrum~\cite{Glabbeek1993} that abstracts from internal computation steps (represented in the transition system by transitions labelled with $\tau$).\notinfull{ We adopt the notation $\bbisimd$ for divergence-preserving branching bisimilarity, and $\bbisim$ for the divergence-insensitive version (see the Appendix for a formal definition)}.
\infull{

In the definition of (divergence-preserving) branching bisimilarity we need the following notation: let $\step{}$ be an $\Atau$-labelled transition relation on a set $\Sta$, and let $a\in\Atau$; we write $s\step{(a)}t$ for ``$s\step{a}t$'' or ``$a=\tau$ and $s=t$''. Furthermore, we denote the transitive closure of $\step{\tau}$ by $\step{}^{+}$ and the reflexive-transitive closure of $\step{\tau}$ by $\step{}^{*}$.

\begin{definition}[Branching Bisimilarity]\label{def:bbisim}
Let $T_1=(\Sta_1,\step{}_1,\uparrow_1)$ and $T_2=(\Sta_2,\step{}_2,\uparrow_2)$ be transition systems. A \emph{branching bisimulation} from $T_1$ to $T_2$ is a binary relation $\R\subseteq\Sta_1\times\Sta_2$ such that for all states $s_1$ and $s_2$, $s_1\R s_2$ implies
\begin{enumerate}
    \item if $s_1\step{a}_1s_1'$, then there exist $s_2',s_2''\in\Sta_2$, s.t. $s_2\step{}_2^{*}s_2''\step{(a)}s_2'$, $s_1\R s_2''$ and $s_1'\R s_2'$;
    \item if $s_2\step{a}_2s_2'$, then there exist $s_1',s_1''\in\Sta_1$, s.t. $s_1\step{}_1^{*}s_1''\step{(a)}s_1'$, $s_1''\R s_2$ and $s_1'\R s_2'$.
\end{enumerate}
The transition systems $T_1$ and $T_2$ are \emph{branching bisimilar} (notation: $T_1\bbisim T_2$) if there exists a branching bisimulation $\R$ from $T_1$ to $T_2$ s.t. $\uparrow_1\R\uparrow_2$.

A branching bisimulation $\R$ from $T_1$ to $T_2$ is \emph{divergence-preserving} if, for all states $s_1$ and $s_2$, $s_1\R s_2$ implies
\begin{enumerate}
\setcounter{enumi}{2}
    \item if there exists an infinite sequence $(s_{1,i})_{i\in\mathbb{N}}$ s.t. $s_1=s_{1,0},\,s_{1,i}\step{\tau}s_{1,i+1}$ and $s_{1,i}\R s_2$ for all $i\in\mathbb{N}$, then there exists a state $s_2'$ s.t. $s_2\step{}^{+}s_2'$ and $s_{1,i}\R s_2'$ for some $i\in\mathbb{N}$; and
    \item if there exists an infinite sequence $(s_{2,i})_{i\in\mathbb{N}}$ s.t. $s_2=s_{2,0},\,s_{2,i}\step{\tau}s_{2,i+1}$ and $s_1\R s_{2,i}$ for all $i\in\mathbb{N}$, then there exists a state $s_1'$ s.t. $s_1\step{}^{+}s_1'$ and $s_1'\R s_{2,i}$ for some $i\in\mathbb{N}$.
\end{enumerate}
The transition systems $T_1$ and $T_2$ are \emph{divergence-preserving branching bisimilar} (notation: $T_1\bbisim^{\Delta}T_2$) if there exists a divergence-preserving branching bisimulation $\R$ from $T_1$ to $T_2$ s.t. $\uparrow_1\R\uparrow_2$.
\end{definition}}

\delete{\todo{Omit the notion of weak trace equivalence presented below. Instead, we should define (later, when we are going to discuss which RTMs are suitable for stream translation) what are the computations associated with a transition system and define the stream translation associated with a transition system similarly as above.}

We also define the weak infinite trace equivalence between transition systems. This relation is useful for classifying the $\omega$-translation realized by transition systems. For any LTS $(\Sta,\step{},\uparrow)$, a state $s\in\Sta$, and an infinite action sequence $\sigma\in\A^{\omega}$, with $\sigma=a_0,a_1,\ldots$, we denote by $s\step{\sigma}$ the fact that there exist $s_0,s_0'\ldots\in\Sta$ such that $s=s_0$, and $s_i\step{}^{*}s_i'\step{a_i}s_{i+1}$ for all $i\geq 0$.

\begin{definition}[Weak Infinite Trace Equivalence]\label{def:wtequ}
For any LTS $T=(\Sta,\step{},\uparrow)$ and state $s\in\Sta$, we define $\witr(s)$ to be the set of weak infinite traces possible from $s$
\begin{equation*}
\witr(s)=\{\sigma\in\A^{\omega}\mid s\step{\sigma}\}
\enskip.
\end{equation*}

Let $T_1=(\Sta_1,\step{}_1,\uparrow_1)$ and $T_2=(\Sta_2,\step{}_2,\uparrow_2)$ be transition systems. We say that $T_1$ and $T_2$ are weak trace equivalent iff $\witr(\uparrow_1)=\witr(\uparrow_2)$, denoted by $T_1\witequ T_2$.
\end{definition}}

\delete{\subsection{Reactive Turing Machine and Executability}

\todo{Integrate the material in this subsection in the previous subsection.}}

The notion of reactive Turing machine (RTM) was put forward in \cite{BLT2013} to mathematically characterise which behaviour is executable by a conventional computing system. We recall the definition of RTMs and the ensued notion of executable transition system. \delete{The definition of RTMs is parameterised with the set $\Atau$, which we  now assume to be a finite set. Furthermore, the definition is parameterised with another finite set $\D$ of \emph{data symbols}. We extend $\D$ with a special symbol $\Box\notin\D$ to denote a blank tape cell, and denote the set $\D\cup\{\Box\}$ of \emph{tape symbols} by $\Dbox$.}
\begin{definition}\label{def:rtm}
A \emph{reactive Turing machine} (RTM) $\M$ is a triple $(\Sta,\step{},\uparrow)$, where
\begin{enumerate}
    \item $\Sta$ is a finite set of \emph{states},
    \item ${\step{}}\subseteq \Sta\times\Dbox\times\Atau\times\Dbox\times\{L,R\}\times\Sta$ is a $(\Dbox\times\Atau\times\Dbox\times\{L,R\})$-labelled \emph{transition relation} (we write $s\step{a[d/e]M}t$ for $(s,d,a,e,M,t)\in{\step{}}$),
    \item ${\uparrow}\in\Sta$ is a distinguished \emph{initial state}.
\end{enumerate}
\end{definition}

Intuitively, the meaning of  a transition $s\step{a[d/e]M}t$ is that whenever $\M$ is in state $s$, and $d$ is the symbol currently read by the tape head, then it may execute the action $a$, write symbol $e$ on the tape (replacing $d$), move the read/write head one position to the left or the right on the tape, and then end up in state $t$.

To formalise the intuitive understanding of the operational behaviour of RTMs, we associate with every RTM $\M$ an $\Atau$-labelled transition system  $\T(\M)$. The states of $\T(\M)$ are the configurations of $\M$, pairs consisting of a state and a tape instance.
\fullversion{We adopt a convention to concisely denote new placement of the tape head marker. Let $\delta$ be an element of $\Dbox^{*}$. Then by $\tphdL{\delta}$ we denote the element of $(\Dbox\cup\check{\Dbox})^{*}$ obtained by placing the tape head marker on the right-most symbol of $\delta$ (if it exists), and $\check{\Box}$ otherwise.
Similarly $\tphdR{\delta}$ is obtained by placing the tape head marker on the left-most symbol of $\delta$ (if it exists), and $\check{\Box}$ otherwise.  A configuration of an ITM consists of pair $(q,\delta)$ of a state and a tape instance. We write $(q,\delta)\step{i/o}_{\I}(q',\delta')$ if ${\step{}_{\I}}(q,d,i)=(q',e,M,o)$, there is an occurrence of the marked version $\check{d}$ of $d$ in $\delta$, $\delta'$ is obtained from $\delta$ by replacing $\check{d}$ by $e$ and replacing the symbol left or right of $\check{d}$ (depending on whether $M=L$ or $M=R$) by its marked version.}

\begin{definition}\label{def:lts-tm}
Let $\M=(\Sta,\step{},\uparrow)$ be an RTM. The transition system $\T(\M)$ \emph{associated with} $\M$ is defined as follows:
\begin{enumerate}
    \item its set of states $\Sta$ consists of the set of all configurations of $\M$;
    \item its transition relation $\step{}$ is the least relation satisfying, for all $a\in\Atau,\,d,e\in\Dbox$ and $\delta_L,\delta_R\in\Dbox^{*}$:
    \begin{itemize}
        \item $(s,\delta_L\check{d}\delta_R)\step{a}(t,\tphdL{\delta_L}e\delta_R)$ iff $s\step{a[d/e]L}t$, and
        \item $(s,\delta_L\check{d}\delta_R)\step{a}(t,\delta_L e\tphdR{\delta_R})$ iff $s\step{a[d/e]R}t$
    \end{itemize}
    ($\tphdL{\delta_L}$ is obtained from $\delta_L$ by placing the tape head marker on the right-most symbol in $\delta_L$, and $\tphdR{\delta_R}$ is obtained analogously from $\delta_R$);
    \item its initial state is the configuration $(\uparrow,\check{\Box})$.
\end{enumerate}
\end{definition}

Turing introduced his machines to define the notion of \emph{effectively computable function} in~\cite{Turing1936}. By analogy, we have a notion of \emph{effectively executable behaviour}~\cite{BLT2013}.

\begin{definition}\label{def:exe}
A transition system is \emph{executable} if it is the transition system associated with some RTM.
\end{definition}
\fullversion{
For the convenience of proofs, we also introduce the RTMs with rules of the form $s\step{a[d/e]N}t$, where $N$ means no movement of the tape head. We denote such machines as RTM$^{N}$. An observation is that one can simulate the transition system of an RTM$^{N}$ by an RTM up to divergence-preserving branching bisimilarity.

\begin{lemma}~\label{lemma:RTMN}
The transition system associated with an RTM$^{N}$ is executable modulo divergence-preserving branching bisimilarity.
\end{lemma}}

%% file: ITM.tex
\section{Executability of Interactive Turing Machines}~\label{sec:itm}

In this section we associate a transition system with every ITM, and
then prove that it is executable modulo divergence-preserving
branching bisimilarity. It is convenient to consider input and output
as separate actions in the transition system associated with an
ITM. We denote by $?i$ the action of inputting the
symbol $i\in\Sigma$, and by $!o$ the action of outputting the symbol $o\in\Sigma$.

\begin{definition}\label{def:lts-itm}
Let $\I=(Q,\step{}_{\I},\qin)$ be an ITM.
The transition system $\T(\I)$ associated with $\I$ is defined as follows:
\begin{enumerate}
\item its set of states is the set
   $\{(s,\delta)\mid s\in Q\cup\{s_o\mid o\in\Sigma_{\lambda}, s\in Q\},\,\delta\mbox{ is a tape instance}\}$;
\item its transition relation $\step{}$ is the least relation
  satisfying, for all $i,o\in\Sigma_{\lambda}$, $d,e\in\Dbox$, and
  $\delta_L,\delta_R\in\Dbox^{*}$:
  \begin{itemize}
  \item
    $(s,\delta_L\check{d}\delta_R)\step{?i}(t_o,\tphdL{\delta_L}e\delta_R)$
      iff $(s,d,i)\step{}_{\I}(t,e,L,o)$ and $i\in\Sigma$,
  \item
    $(s,\delta_L\check{d}\delta_R)\step{?i}(t_o,\delta_L
    e\tphdR{\delta_R})$ iff $(s,d,i)\step{}_{\I}(t,e,R,o)$ and $i\in\Sigma$,
 \item
    $(s,\delta_L\check{d}\delta_R)\step{\tau}(t_o,\tphdL{\delta_L}e\delta_R)$
      iff $(s,d,i)\step{}_{\I}(t,e,L,o)$ and $i=\lambda$,
  \item
    $(s,\delta_L\check{d}\delta_R)\step{\tau}(t_o,\delta_L
    e\tphdR{\delta_R})$ iff $(s,d,i)\step{}_{\I}(t,e,R,o)$ and $i=\lambda$,
   \item
    $(s_o,\delta)\step{!o}(s,\delta)$ iff $o\in\Sigma$, and
    $(s_o,\delta)\step{\tau}(s,\delta)$ iff $o=\lambda$.
  \end{itemize}
\item its initial state is the configuration $(\qin,\check{\Box})$.
\end{enumerate}
\end{definition}

The following theorem shows that every transition systems associated
with an ITM can be simulated by an RTM.\notinfull{ (A proof of the
  theorem is included in the Appendix.)}
\infull{
In the proof it is convenient to allow RTMs to have transitions of the form $s\step{a[d/e]S}t$, where $S$ is a stay transition with no movement of the tape head. We refer to such machines as RTMs with stay transitions. The operational semantics of RTMs can be extended to an operational semantics for RTMs with stay transitions by adding the clause: $(s,\delta_L\check{d}\delta_R)\step{a}(t,\delta_L \check{e}\delta_R)$ iff $s\step{a[d/e]S}t$. The transition system of an RTM with stay transitions can be simulated by an RTM up to divergence-preserving branching bisimilarity.

\begin{lemma}~\label{lemma:RTMN}
The transition system associated with an RTM with stay transitions is executable up to divergence-preserving branching bisimilarity.
\end{lemma}

\begin{proof}
We suppose that $\M=(\Sta,\step{},\uparrow)$ is an RTM with stay transitions, and its transition system is $\T(\M)$. We define a normal RTM $\M'=(\Sta_1,\step{}_1,\uparrow_1)$ that simulates $\T(\M)$ as follows:

\begin{enumerate}
\item $\Sta_1=\Sta\cup\{s_t\mid s,t\in\Sta\}$;
\item $s\step{a[d/e]L}_1 t$ iff $s\step{a[d/e]L}t$;
\item $s\step{a[d/e]R}_1 t$ iff $s\step{a[d/e]R}t$;
\item $s\step{a[d/e]L}_1 s_t$ and $s_t\step{\tau[d/d]R}_1t$ iff $s\step{a[d/e]S}t$; and
\item $\uparrow_1=\uparrow$.
\end{enumerate}

Then it is straight forward to $\T(\M')\bbisimd\T(\M)$.
\end{proof}
}

\begin{theorem}\label{thm:itm-lts-iso}
For every ITM $\I$ there exists an RTM $\M$, such that $\T(\I)\bbisimd \T(\M)$.
\end{theorem}
\infull{
    We let $\I=(Q,\step{}_{\I},\qin)$ be an ITM. By Lemma~\ref{lemma:RTMN}, it is enough to show that there exists an RTM with stay transitions $\M$ satisfying $\T(\M)\bbisimd\T(\I)$.
    We construct $\M=(\Sta,\step{},\uparrow)$ as follows:
    \begin{enumerate}
        \item $\Sta=\In\cup\Out$, where $\In=Q$ and $\Out=\{s_o\mid o\in\Sigma_{\lambda}, s\in Q\}$;
        \item the transition relation $\step{}$ is defined by:
        $s\step{\mathit{in}(i)[d/e]M}t_o$ if $(s,d,i)\step{}_{\I}(t,e,M,o)$, and $s_o\step{\mathit{out}(o)[e/e]S}s$ for all $s\in\Sta$, $o\in\Sigma_{\lambda}$; and
        \item $\uparrow=\qin$.
    \end{enumerate}
    Then according to Definitions~\ref{def:lts-tm} and~\ref{def:lts-itm}, we get a transition system $\T(\M)=\T(\I)$, where `=' is the pointwise equality, which also implies $\T(\M)\bbisimd \T(\I)$.
}
As a consequence we have the following corollary.

\begin{corollary}\label{cor:itm-lts-exe}
The transition system associated with an ITM is executable modulo divergence-preserving branching bisimilarity.
\end{corollary}

%% file: ExeITM.tex
\section{Executable $\omega$-Translations}\label{sec:exeomega}
 Recall that an $\omega$-translation is defined to be interactively computable if, and only if, it can be realised by an ITM. RTMs are designed for exhibiting the expressive power of executable transition systems, rather than $\omega$-translations, and not every RTM naturally has an $\omega$-translation associated with it. Imposing some restrictions on the formalism of RTMs, however, we shall define a subclass of RTMs with which an $\omega$-translation is naturally associated. The $\omega$-translation realised by such an RTM is then called \emph{executable}, and we shall establish that an $\omega$-translation is interactively computable if, and only if, it is executable.

By analogy to the systems described in the theory of interactive computation, we let the RTMs for $\omega$-translations execute in steps, in such a way that with every step a pair of input and output actions can be associated. With every infinite computation of the RTM we can then associate a interaction pair, and the RTM will thus give rise to an $\omega$-translation.

 \begin{definition}~\label{def:rtm-omega}
 Let $\Atau=\{?i,!o\mid i,o\in\{0,1\}\}\cup\{\tau\}$, and let $\M=(\Sta,\step{},\uparrow)$ be an RTM with $\Atau$ as its set of labels. Then $\M$ is an \emph{RTM for $\omega$-translation} if it satisfies the following properties:
 \begin{enumerate}
 \item the set of states $\Sta$ is partitioned into disjoint sets $\in$ of input states and $\Exe$ of execution states, i.e., $\Sta=\In\cup\Exe$ and $\In\cap\Exe=\emptyset$;
 \item the initial state ${\uparrow}$ is an input state, i.e., ${\uparrow}\in\In$;
 \item for a transition $s\step{a[d/e]M}t$, if $s\in\In$, then $a\in \{?0,?1\}$ and $t\in\Exe$; if $s\in\Exe$, then $a\in\{!0,!1,\tau\}$ and $t\in \In$; and
 \item for all $(s,d)\in\Exe\times\Dbox$, there is at most one transition of the form $s\step{a[d/e]M}t$; and
 \item for all $(s,d)\in\In\times\Dbox$, there are exactly two transitions of the form $s\step{a[d/e]M}t$, one with $a=?0$ and one with $a=?1$.
 \end{enumerate}
 \end{definition}
 In the following lemma we establish some properties of the transition system associated with an RTM for $\omega$-translation. \notinfull{(See the Appendix for a proof of the lemma.)}

\begin{lemma}\label{lemma:i-o-lts}
Let $\M$ be an RTM for $\omega$-translation. Then $\T(\M)=(\Sta_{\M},\step{}_{\M},\uparrow_{\M})$ satisfies the following properties:
\begin{enumerate}
    \item (\emph{Alternation}) The set of states $\Sta_{\M}$ is partitioned into a set of input states $\In_{\M}$ and a set of output states $\Exe_{\M}$, i.e., $\Sta_{\M}=\In_{\M}\cup\Exe_{\M}$ and $\In_{\M}\cap\Exe_{\M}=\emptyset$. For every transition $s\step{a} s'$, if $s\in\In_{\M}$, then $a\in\{?0,?1\}$ and $s'\in\Exe_{\M}$; if $s\in\Exe_{\M}$, then $a\in\{!0,!1,\tau\}$ and $s'\in\In_{\M}$.
    \item (\emph{Unambiguity}) For every $s\in\Exe_{\M}$, there is exactly one outgoing transition $s\step{a} s'$ with $a\in\{!0,!1,\tau\}$.
    \item (\emph{Totality}) For every $s\in\In_{\M}$, there are exactly two outgoing transitions, labelled with $?0$ and $?1$, respectively.
\end{enumerate}
\end{lemma}
\infull{
\begin{proof}
 A state in $\Sta_{\M}$ is a configuration $(s,\delta)$ of $\M$, and we can make a partition of the set of all configurations according to the control states. If $s\in\In$, then $(s,\delta)\in\In_{\M}$; if $s\in\Exe$, then $(s,\delta)\in\Exe_{\M}$, where $\In$ and $\Exe$ are defined in Definition~\ref{def:rtm-omega}.
\begin{enumerate}
    \item (Alternation) By condition 1 in Definition~\ref{def:rtm-omega}, we have $\Sta=\In\cup\Exe$ and $\In\cap\Exe=\emptyset$, which infers $\Sta_{\M}=\In_{\M}\cup\Exe_{\M}$ and $\In_{\M}\cap\Exe_{\M}=\emptyset$; moreover, by condition 2, for a transition $s\step{a[d/e]M}t$, if $s\in\In$, then $a\in \{?0,?1\}$ and $t\in\Exe$; if $s\in\Exe$, then $a\in\{!0,!1,\tau\}$ and $t\in \In$, which infers that for every transition $s\step{a} s'$, if $s\in\In_{\M}$, then $a\in\{?0,?1\}$ and $s'\in\Exe_{\M}$; if $s\in\Exe_{\M}$, then $a\in\{!0,!1,\tau\}$ and $s'\in\In_{\M}$.
    \item (Unambiguity) By condition 3 in Definition~\ref{def:rtm-omega}, for all $(s,d)$ where $s\in\Exe$ and $d\in\Dbox$, there is at most one transition $s\step{o[d/e]M}t$, which infers that for every $s\in\Exe_{\M}$, there is exactly one outgoing transition $s\step{a} s'$ with $a\in\{!0,!1,\tau\}$.
    \item (Totality) By condition 4 in Definition~\ref{def:rtm-omega}, for all $(s,d)$ where $s\in\In$ and $d\in\Dbox$, there are exactly two transitions of the form $s\step{i[d/e]M}t$, with $?0$ and $?1$ as there labels, respectively, which infers that for every $s\in\In_{\M}$, there are two outgoing transitions labelled by $?0$ and $?1$, respectively.
\end{enumerate}
\end{proof}
}

We call a transition that satisfies the conditions of Lemma~\ref{lemma:i-o-lts} an \emph{i/o transition system}.
Moreover, by analogy to the interactiveness condition for ITMs, we impose an interactiveness condition on RTMs for $\omega$-translation.

\begin{definition}
An i/o transition system is interactive, if
  for every $s\in\Sta$ and $s\step{?i}s_0$ with $i\in\{0,1\}$, and for every sequence $s_0\step{}s_1\step{}\cdots$, there exists a natural number $i$, such that $s_i\step{!o}s_{i+1}$ with $o\in\{0,1\}$.

An RTM for $\omega$-translation is \emph{interactive} if the associated i/o transition system is.
\end{definition}
We define the $\omega$-translation realized by an RTM by defining the $\omega$-translation realized by the i/o transition system associated with it. Let $\T=(\Sta,\step{},\uparrow)$ be an i/o transition system, let $s\in\Sta$, and let $\sigma\in\A^{\omega}$, say $\sigma=a_0,a_1,\ldots$; we write $s\step{\sigma}$ if there exist
  $s_0,s_0', s_1,s_1',\ldots\in\Sta$
such that $s=s_0$, and $s_i\step{}^{*}s_i'\step{a_i}s_{i+1}$ for all $i\geq 0$. (By $\step{}^{*}$ we denote the reflexive-transitive closure of the relation $\step{\tau}$.)
If $\sigma\in \mathcal{A}^\omega$ and $s \step{\sigma}$, then $\sigma$ is a \emph{weak infinite trace}
from $s$. We denote by $\witr(s)$ the set of weak infinite traces from $s$.\delete{, i.e.,

\begin{equation*}
\witr(s)=\{\sigma\in\A^{\omega}\mid s\step{\sigma}\}
\enskip.
\end{equation*}
}
\begin{definition}\label{def:in-out-lts}
 Let $\T$ be an i/o transition system, and $s_0$ be the initial state. For $\sigma\in\witr(s_0)$,
 the \emph{input stream realised by $\sigma$} is the stream $\vec{x}\in \Sigma^{\omega}$ such that $\vec{x}=x_1x_2\ldots$, where $x_j=i$ if $?i$ is the $j$-th input action in $\sigma$, and similarly for the \emph{output stream realized by $\sigma$}.
 \delete{
 we define $(\vec{x},\vec{y})\in \{?0,?1\}^{\omega}\times \{!0,!1\}^{\omega}$ as the pair of input and output streams realized by $\sigma$ as follows.
\begin{enumerate}
    \item Its input stream is $\vec{x}=x_1x_2\ldots$, where $x_j=i$, if $?i$ is the $j$-th input action in $\sigma$, and
    \item its output stream is $\vec{y}=y_1y_2\ldots$, where $y_j=o$, if $!o$ is the $j$-th output action in $\sigma$.
\end{enumerate}
}%
We say that $\T$ realizes $\omega$-translation $\phi:\Sigma^{\omega}\rightarrow \Sigma^{\omega}$ iff, for every $\vec{x}\in \Sigma^{\omega}$, there exists a trace $\sigma\in\witr(s_0)$ with $\vec{x}$ as its input stream, and for every such trace, its output stream is $\vec{y}=\phi(\vec{x})$.
\delete{
Moreover, if $\T$ is an executable i/o transition system, we say that it realizes an executable $\omega$-translation.
}
\end{definition}

We can now define when an $\omega$-translation is executable.

\begin{definition}\label{def:executable-omega}
An $\omega$-translation is executable if it can be realized by an executable i/o transition system.
\end{definition}

The following lemma establishes that an $\omega$-translation can be associated with every interactive i/o transition system.
\begin{lemma}~\label{lemma:io-LTS-omega}
    If an i/o transition system is interactive, then it realises an $\omega$-translation.
\end{lemma}
\infull{
\begin{proof}
   Let $\T$ be an i/o interactive transition system, and let $s_0$ be the initial state of $\T$. By Definition~\ref{def:in-out-lts}, we need to show that there exists an $\omega$-translation $\phi$ such that for every $\vec{x}\in\Sigma^{\omega}$, there exists a trace $\sigma\in\witr(s_0)$ with input stream $\vec{x}$, and for every trace with input stream $\vec{x}$, its output stream is $\vec{y}=\phi(\vec{x})$.

   By the alternation condition in Lemma~\ref{lemma:i-o-lts}, every $\sigma\in\witr(s_0)$ is of the form $i_0 o_0 i_1 o_1\ldots$ where $i_j\in\{?0,?1\}$ and $o_j\in\{!0,!1,\tau\}$.
   Let $\vec{x}$ be an arbitrary input stream, by the totality condition in Lemma~\ref{lemma:i-o-lts}, we can find a trace $\sigma\in\witr(s_0)$ with input stream $\vec{x}$.

    Moreover, given an trace $\sigma$ with an infinite input stream $\vec{x}$, by interactiveness, it would always produce an infinite output stream $\vec{y}$.

    Finally, by unambiguity, there do not exist two traces sharing the same input stream. It follows that for every trace with input stream $\vec{x}$, its output stream is $\vec{y}$. Hence, we relate with every input stream a unique output stream, in a way, we get a $\omega$-translation from $\T$.
\end{proof}
}

\infull{
It is not hard to show the following lemmas,

\begin{lemma}\label{lemma:bis-translation}
Let $\T_1$ and $\T_2$ be two i/o transition systems, and $\T_1\bbisim \T_2$. Then they realize the same $\omega$-translation.
\end{lemma}

\begin{proof}
We let $s_1$ and $s_2$ be the initial states of $\T_1$ and $\T_2$, respectively. As $\T_1\bbisim \T_2$, we have that for every $\sigma\in\witr(s_1)$, there exists a trace $\sigma'\in\witr(s_2)$, and they share the same input and output stream, and vice versa. It follows that $\T_1$ and $\T_2$ realize the same $\omega$-translation.
\end{proof}

\begin{lemma}~\label{lemma:finite-input}
Let $\T$ be an interactive i/o transition system, and let $s_0$ be its initial state, then the following function is computable:
$g: \Sigma^{*}\rightarrow \Sigma^{*}$, satisfying that if $g(x)=y$, then for every $\sigma\in\witr(s_0)$ with input and output stream $\vec{x}$ and $\vec{y}$, if $x\prec\vec{x}$, then $y\prec\vec{y}$.
\end{lemma}

\begin{proof}
We consider a finite trace from $s_0$, we can associate with such a trace its input and output sequences in a similar way as defined in Definition~\ref{def:in-out-lts}.
By Lemma~\ref{lemma:i-o-lts}, there is only one finite trace with $x$ as its input sequence, and its output sequence is $y$. By totality, it holds for every $x\in\Sigma^{*}$.
As the transition relation of i/o transition systems are computable, $g$ is also computable.
\end{proof}
}

Moreover, we have the following theorem. \notinfull{(A proof can be found in the appendix.)}
\begin{theorem}~\label{thm:lc-exe-translation2}
 An $\omega$-translation is an executable iff it is a limit-continuous total function.
\end{theorem}
\infull{
\begin{proof}
We let $\phi$ be an $\omega$-translation.
\begin{enumerate}
\item For the ``only if'' part, we need to show that there exists a computable total function $g:\Sigma^{*}\rightarrow\Sigma^{*}$, such that $g$ is monotonic and for all strictly increasing chains $u_1\prec u_2\prec\ldots\prec u_t\prec\ldots$ with $u_t\in\Sigma^{*}$ ($t\geq 1$), one has $\phi(\lim_{t\rightarrow\infty}u_t)=\lim_{t\rightarrow\infty}g(u_t)$.

    We assume that $\phi$ is realized by an interactive i/o transition system $\T$, and we let $s_0$ be the initial state of $\T$. By Lemma~\ref{lemma:finite-input} the following function is computable: $g: \Sigma^{*}\rightarrow \Sigma^{*}$, satisfying that if $g(x)=y$, then for every $\sigma\in\witr(s_0)$ with input and output stream $\vec{x}$ and $\vec{y}$, if $x\prec\vec{x}$, then $y\prec\vec{y}$. By unambiguity and totality, $g$ is a monotonic and total computable function.

 Moreover, for a strictly increasing chain $u_1\prec u_2\prec\ldots\prec u_t\prec\ldots$ with $u_t\in\Sigma^{*}$ for $t\geq 1$, the computation of $\lim_{t\rightarrow\infty}g(u_t)$ is the execution of a trace $\sigma$ receiving the input stream $\lim_{t\rightarrow\infty}u_t$. Hence we have $\phi(\lim_{t\rightarrow\infty}u_t)=\lim_{t\rightarrow\infty}g(u_t)$.

    Thus, $g$ is the computable total function we need, and it follows that $\phi$ is a computable limit-continuous total function.

\item

    For the ``if'' part, we assume that $\phi$ is a total limit-continuous function, and design an RTM $\M$ to realize this translation.
    By Theorem~\ref{thm:int-trans}, $\phi$ is interactively computable by some ITM $\M'$. According to Definition~\ref{def:lts-itm} and Lemma~\ref{lemma:i-o-lts}, the transition system associated with $\M'$ is an i/o transition system, moreover, according to Corollary~\ref{cor:itm-lts-exe}, it is an executable i/o transition system. Therefore, we have shown that $\phi$ is an executable $\omega$-translation by Lemma~\ref{lemma:bis-translation}.
\end{enumerate}
\end{proof}
}

By Theorem~\ref{thm:int-trans}, we have the following corollary.
\begin{corollary}~\label{cor:equiv}
An $\omega$-translation is executable iff it is interactively computable.
\end{corollary}
Therefore, the classes of computable limit-continuous functions, interactively computable $\omega$-translations and executable $\omega$-translations coincide.

%% file: Advice.tex
\section{Advice}\label{sec:advice}

In~\cite{vLW2001}, the computational power of evolving interactive systems is studied using ITMs. Particularly, a mechanism called \emph{advice function} is introduced to enhance the computational power of an ITM. In this way, the insertion of external information into the course of a computation is allowed, which leads to a non-uniform operation. In this section, we introduce the notion of advice as a process in parallel composition with an RTM, and show that advice processes indeed give the systems more expressive power. 
\delete{
\subsection{Advice Process}}

In this section, we consider advices as functions over natural numbers. In order to record a number on the tape, a natural number $n$ is encoded by a sequence $n$ ``$1$''s ending with a ``$0$''. In~\cite{vLW2001}, the notion of ITM with advice is defined as follows.

\begin{definition}
An \emph{advice} function is a function $f: \mathbb{N}\rightarrow \mathbb{N}$. An ITM with advice (ITM/A) is equipped with a separate \emph{advice tape} and a distinguished \emph{advice state}. By writing the value of the argument $x$ on the advice tape and by entering into the advice state, the value of $f(x)$ will appear on the advice tape in a single step. By this action, the original contents of the advice tape is completely overwritten.
\end{definition}

Here we do not put the restriction on the length of the advice function as in~\cite{vLW2006}, since it does not make a difference in the issue of computability, and we are not yet interested in the issue of complexity. It is obvious that ITMs with uncomputable advice functions cannot be simulated by any RTM, as uncomputable advice function cannot be evaluated by the mechanism of RTMs. As an extension, we equip RTMs with advice processes which enable the simulation of ITM/As.

An advice process $A_f$ is designed to compute the function $f$, and can only interact with a certain RTM $\M$. As an advice function is not necessarily computable, we cannot associate with every advice process an executable transition system. An RTM $\M$ communicates with $A_f$ as follows: when it needs to get the result of $f(i)$, it enters a special control state $a_f$, and starts to send  a sequence of $i$ ``1'' s and a ``0'' , which is already written on the tape, to the channel $\overline{\mathit{in}}$, and then, it receives the result sequence $f(i)$ ``1''s and a ``$0$'' from $\mathit{out}$ channel, and write them on the tape. This procedure ends up with another control state. We can model an advice process as follows.

\begin{definition}\label{def:advice}
Let $f: \mathbb{N}\rightarrow \mathbb{N}$ be a function, $A_f$ is an advice process for $f$ with transition system $\T(A_f)=(\Sta,\rightarrow,\uparrow)$, where
\begin{enumerate}
    \item $\Sta=\{s_i\mid i=0,1,2,\ldots\}\cup\{t_i\mid i=0,1,2,\ldots\}$, and
    \item $s_i\step {\mathit{in}? 1} s_{i+1},\, i=0,1,2\ldots\quad s_i\step{\mathit{in}? 0}t_{f(i)},\, i=1,2\ldots$\\
        $t_i\step{\mathit{out} !1} t_{i-1},\, i=1,2\ldots\quad t_0\step{\mathit{out}! 0} s_0$
    \item $\uparrow=s_0$.
\end{enumerate}
\end{definition}

The behaviour of $A_f$ is deterministic. It receives a sequence of $i$ ``$1$''s from the channel $\mathit{in}$, followed by a ``$0$'' symbol, indicating the end of the sequence, and then, it produces $f(i)$ ``$1$''s to the channel $\mathit{out}$, also followed by a ``$0$'' symbol. This procedure is repeated indefinitely.

 The parallel composition of an RTM $\M$ and an advice process $A_f$, we write as $[\M\parallel A_{f}]_{\C}$. The parallel composition is defined in the same way as the parallel composition of two RTMs in~\cite{BLT2013}, where $\C=\{\mathit{in},\mathit{out}\}$ is the set of restricted names for communication. If $\M$ is an RTM and $A_f$  is an advice process, then we call $[\M\parallel A_f] _{\C}$ a reactive Turing machine with advice (RTM/A).

Note that, since advice functions and advice processes have the same computational power, by Corollary~\ref{cor:equiv}, an $\omega$-translation is realisable by an ITM/A if, and
only if, it is realisable by an RTM/A.
\delete{
\subsection{Executability with Advice}}

    \delete{In the theory of executability, it is substantial to figure out the expressive power of the labelled transition systems associated with RTM/As.
    We now proceed to show that every boundedly branching labelled transition system can be simulated by some RTM/A up to divergence-preserving branching bisimilarity, providing that the advice is not restricted to evaluate computable functions.}

Let $\T$ be any bounded branching transition system (not necessarily effective). Based on a presupposed
encoding of its sets of states and actions and its transition relation, let the advice function $f_{\T}$ be such
that for the code of a state it yields the code of the set of all outgoing transitions of that state. It is
straightforward to define an RTM that simulates $\T$ with the help of $f_{\T}$. Then we obtain the following result.

\begin{theorem}~\label{thm:bound-lts-rtma}
If $\T$ is a boundedly branching labelled transition system, then there exists an RTM/A $[\M\parallel A_{f}]_{\C}$ such that $\T([\M\parallel A_{f}]_{\C})\bbisimd \T$.
\end{theorem}
\infull{
\begin{proof}
We assume that $\T=(\Sta_{\T},\step{}_{\T},\uparrow_{\T})$ is an $\Atau$-labelled transition system. It has $n$ distinct action labels and its branching degree is bounded by $k$.
Then we encode $\Atau$ and $\Sta_{\T}$ as natural numbers. Let $\encode{a}$ and $\encode{s}$ be the encodings of an action and a state, and $\encode{x_1,x_2,\ldots,x_n}$ be the encoding of an $n$-tuple.

The advice process $A_{f}$ realizes the following function:
\begin{equation*}
f(\encode{s})=\encode{a_1,\ldots, a_m, s_1,\ldots,s_m}
\enskip,
\end{equation*}
where $(a_i,s_i)\in \{(a_1,s_1),\ldots,(a_m,s_m)\}$ iff $s\step{a_i}_{\T} s_i$.

An outline of the execution of $\M$ is defined as follows.
\begin{enumerate}
    \item We need the following control states: $\mathit{initial}$, $\mathit{advice}$, $\mathit{decode}$, $\mathit{next}_{\Atau^{\leq k}}$ ($\Atau^{\leq k}$ ranges over all $\Atau$ words with at most length $k$), $\mathit{choose}_i$ ($i\leq k$).
    \item The execution of $\M$ is as follows, its initial configuration is $(\mathit{initial},\Box)$.
    \begin{enumerate}
        \item In $\mathit{initial}$ state, the machine writes the encoding of initial state of the transition system $\encode{\uparrow_{\T}}$ on the tape, and reaches $\mathit{advice}$ state.
        \begin{equation*}
         (\mathit{initial}, \Box)\step{}^{*} (\mathit{advice}, \encode{\uparrow_{\T}})
         \enskip.
        \end{equation*}
        \item In $\mathit{advice}$ state, the machine sends the encoding of the current state $\encode{s_0}$ to the advice process, and gets the encoding of list of all possible transitions $\encode{a_1,\ldots, a_m, s_1,\ldots,s_m}$ from the advice process.
        \begin{equation*}
         (\mathit{advice},\encode{s_0})\step{}^{*} (\mathit{decode},\encode{a_1,\ldots, a_m, s_1,\ldots,s_m})
         \enskip.
        \end{equation*}
        \item In $\mathit{decode}$ state, the machine decodes all the actions from the tape, and enters one of the $\mathit{next}$ state.
        \begin{equation*}
        (\mathit{decode},\encode{a_1,\ldots, a_m, s_1,\ldots,s_m})\step{}^{*} (\mathit{next}_{\{a_1,\ldots,a_m\}},\encode{s_1,\ldots,s_m})
        \enskip.
        \end{equation*}
        \item In $\mathit{next}_{\{a_1,\ldots,a_m\}}$ state, the machine chooses one of the actions. For every $i=1,\ldots,m$, there is a transition
        \begin{equation*}
        (\mathit{next}_{\{a_1,\ldots,a_m\}},\encode{s_1,\ldots,s_m})\step{a_i} (\mathit{choose}_i,\encode{s_1,\ldots,s_m})
        \enskip.
        \end{equation*}
        \item In $\mathit{choose}_i$ state, the machine projects the encoding $\encode{s_1,\ldots,s_m}$ to the encoding of the $i$-th state, and enters $\mathit{advice}$ state again.
        \begin{equation*}
        (\mathit{choose}_i,\encode{s_1,\ldots,s_m})\step{}^{*}(\mathit{advice},\encode{s_i})
        \enskip.
        \end{equation*}
    \end{enumerate}
\end{enumerate}
The above procedure describes the simulation of a step of transition $s_0\step{a_i}_{\T}s_i$ in $\T$. Note that the choice of the transition is happened only in the state $\mathit{next}_{\{a_1,\ldots,a_m\}}$. Moreover, no infinite $\tau$-transition sequence is introduced for simulation.
Hence, we are able to verify that $\T([\M\parallel A_{f}]_{\C})\bbisimd \T$.
\end{proof}
}

If we, instead, let the advice function $f_{\T}$ be such that on the code of a pair of a state $s$ and a natural
number $i$ yields the code of the $i$th outgoing transition of $s$, then we can extend the simulation to transition
systems with countable many states and transitions.

\begin{theorem}\label{thm:lts-rtma}
If $T$ is a countable labelled transition system, then there exists an RTM/A $[\M\parallel A_{f}]_{\C}$ such that $\T([\M\parallel A_{f}]_{\C})\bbisim T$.
\end{theorem}
\infull{
\begin{proof}
We assume that $\T=(\Sta_{\T},\step{}_{\T},\uparrow_{\T})$ is a countable $\Atau$-labelled transition system. It has $n$ distinct action labels and it possibly has infinitely branching.
Then we encode $\Atau$ and $\Sta_{\T}$ as natural numbers. Let $\encode{a}$ and $\encode{s}$ be the encodings of an action and a state, and $\encode{x_1,x_2,\ldots,x_n}$ be the encoding of an $n$-tuple.

The transition relation $\step{}_{\T}$ maps a state, namely, $s_0$, to a possibly infinite set $\{(a_i,s_i)\mid s_0\step{a_i}_{\T} s_i\}$, denoted by $s_0\step{}_{\T}$. We define an order $<_{\T}$ over the elements in the set $s_0\step{}_{\T}$ such that $(a,s)<_{\T}(a',s')$, if $\encode{a,s}<_{\T}\encode{a',s'}$.

The advice function $A_{f}$ realizes the following function:
\begin{equation*}
f(\encode{s_0,i})=\encode{a_i,s_i}
\enskip,
\end{equation*}
where $(a_i,s_i)$ is the $i$-th element from $s_0\step{}_{\T}$ regarding to $<_{\T}$.

An outline of the execution of $\M$ is defined as follows.
\begin{enumerate}
    \item We need the following control states: $\mathit{initial}$, $\mathit{advice}$, $\mathit{decode}$, $\mathit{next}_{\Atau}$, $\mathit{choose}_i$ ($i=1,2$).
    \item The execution of $\M$ is as follows, we use a pair $(s,\delta)$ to denote the current configuration of the machine.
    \begin{enumerate}
        \item In $\mathit{initial}$ state, the machine writes the encoding of the initial state of the transition system $\encode{\uparrow_{\T}}$ on the tape, and reaches $\mathit{advice}$ state.
        \begin{equation*}
         (\mathit{initial}, \Box)\step{}^{*} (\mathit{advice}, \encode{\uparrow_{\T},1})
         \enskip.
        \end{equation*}
        \item In $\mathit{advice}$ state, the machine either increase the counter $i$ by $1$, or sends $\encode{s_0,i}$ to the advice, and gets $\encode{(a_i,s_i)}$ from the advice.
        \begin{eqnarray*}
         (\mathit{advice},\encode{s_0,i})\step{}^{*}(\mathit{advice},\encode{s_0,i+1}), or\\
         (\mathit{advice},\encode{s_0,i})\step{}^{*} (\mathit{decode},\encode{s_0,s_i,a_i})
         \end{eqnarray*}
        \item In $\mathit{decode}$ state, the machine decodes the action $a_i$ from the tape, and enters the state $\mathit{next_{a_i}}$.
        \begin{equation*}
        (\mathit{decode},\encode{s_0,s_i,a_i})\step{}^{*} (\mathit{next}_{a_i},\encode{s_0,s_i})
        \enskip.
        \end{equation*}
        \item In $\mathit{next}_{a_i}$ state, the machine either performs the action, or change its current choice to another transition.
        \begin{eqnarray*}
        (\mathit{next}_{a_i},\encode{s_0,s_i})\step{\tau} (\mathit{choose}_1,\encode{s_0,s_i}), or\\
        (\mathit{next}_{a_i},\encode{s_0,s_i})\step{a_i} (\mathit{choose}_2,\encode{s_0,s_i})
        \enskip.
        \end{eqnarray*}
        \item In $\mathit{choose}_i$ state (i=1,2), the machine projects the encoding $\encode{s_1,s_2}$ to the encoding of the $i$-th state, and enters $\mathit{advice}$ state again.
        \begin{equation*}
        (\mathit{choose}_i,\encode{s_1,s_2})\step{}^{*}(\mathit{advice},\encode{s_i,1})
        \enskip.
        \end{equation*}
    \end{enumerate}
\end{enumerate}
One can verify that $\R=\{(s,s')\mid s\in\Sta_{\T}, s'=\\(\mathit{advice}, \encode{s,i})\mbox{ or } (\mathit{decode},\encode{s,a_i,s_i}) \mbox{ or } (\mathit{next}_{a_i},\encode{s,s_i})\mbox{ or } (\mathit{choose}_1,\encode{s,s_i})\mbox{ or } (\mathit{choose}_2,\encode{s_i,s})\}$ is a branching bisimulation relation. Hence, we have $\T([\M\parallel A_{f}]_{\C})\bbisim T$.
\end{proof}
}
 Note that the transition system associated with an RTM/A is boundedly branching. Hence, by Theorem 2 in~\cite{LY14}, if a transition system has no divergence up to $\bbisimd$ and is unboundedly branching up to $\bbisimd$, then it is not executable modulo $\bbisimd$. It follows that there exist countable unboundedly branching transition systems that cannot be simulated by an RTM/A modulo $\bbisimd$.

\delete{Actually, for unboundedly branching transition systems with no divergence, it is unavoidable to introduce divergence in simulation. We first have the following fact.

\begin{lemma}~\label{lemma:branching-itma}
The labelled transition system associated with an RTM/A is boundedly branching.
\end{lemma}

Then we recall the following lemma in~\cite{LY14}.

 \begin{lemma}\label{lemma:divergence}
If a transition system is boundedly branching and does not have divergence up to $\bbisimd$, then it is boundedly branching up to $\bbisimd$.
\end{lemma}

 By analogy to the analysis in~\cite{LY14}, we have the following result.

\begin{corollary}\label{cor:unboundly-lts-itma}
There exists an unboundedly branching labelled transition system $T$, such that there is no RTM/A $[\M\parallel A_{f}]_{\C}$, satisfying $\T([\M\parallel A_{f}]_{\C})\bbisimd T$.
\end{corollary}
}

%% file: Conclusion.tex
\section{Conclusion}\label{sec:conclusion}

We have discussed the relationship between two models of computation that take interaction into account. We have established that the model of RTMs subsumes and is more expressive the model of ITMs when it comes specifying behaviour, and coincides with the model of ITMs when it comes to defining $\omega$-translations.

Furthermore, we have shown that RTMs admit an extension with advice that facilitates modelling non-uniform behaviour. In \cite{BLT2013} it was established that every effective transition system can be simulated by an RTM. Our result that every countable transition system can be simulated by an RTM with advice further confirms the universal expressiveness of the notion of RTM.

In~\cite{Verbaan2006}, a complexity theory for interactive computation has been defined on the basis of ITMs and $\omega$-translations. Clearly, such a complexity theory could also be based on the restricted class of RTMs for $\omega$-translation. Such a complexity theory could then further be generalised towards a complexity theory for general executable behaviour. 

%% file: Appendix.tex
\newpage
\section*{Appendix}\label{appendix}

\myparagraph{Definition of limit-continuous functions}

Van Leeuwen and Wiedermann present in~\cite{vLW2006} a characterisation of the interactively computable $\omega$-translations by showing that they can be approximated by classically computable partial functions on finite sequences over $\Sigma$.
For finite and infinite sequences $\vec{x}$ and $\vec{y}$, we write $\vec{x}\prec\vec{y}$ if $\vec{x}$ is a finite and strict prefix of $\vec{y}$, and $\vec{x}\preceq\vec{y}$ if $\vec{x}\prec\vec{y}$ or $\vec{x}=\vec{y}$. We use the following definition of monotonic functions and limit-continuous functions.
\begin{definition}\label{def:mono-limit-cont}
\begin{enumerate}
\item A partial function $f:\Sigma^{*}\rightharpoonup\Sigma^{*}$ is \emph{monotonic} if for all $\vec{x},\vec{y}\in\Sigma^{*}$ such that $\vec{x}\prec\vec{y}$ and $f(\vec{y})$ is defined, it holds that $f(\vec{x})$ is defined as well and $f(\vec{x})\preceq f(\vec{y})$.
\item A partial function $\phi: \Sigma^{\omega}\rightarrow\Sigma^{\omega}$ is called \emph{limit-continuous} if there exists a classically computable monotonic partial function $f:\Sigma^{*}\rightarrow\Sigma^{*}$ such that
  $\phi(\lim_{k\rightarrow\infty}\vec{x}_k)=\lim_{k\rightarrow\infty}f(\vec{x}_k)$
for all strictly increasing chains $\vec{x}_1\prec \vec{x}_2\prec\cdots\prec \vec{x}_k\prec\cdots$ with $\vec{x}_k\in\Sigma^{*}$.
\end{enumerate}
\end{definition}

\myparagraph{Definition of Branching Bisimilarity}

In the definition of (divergence-preserving) branching bisimilarity we need the following notation: let $\step{}$ be an $\Atau$-labelled transition relation on a set $\Sta$, and let $a\in\Atau$; we write $s\step{(a)}t$ for ``$s\step{a}t$'' or ``$a=\tau$ and $s=t$''. Furthermore, we denote the transitive closure of $\step{\tau}$ by $\step{}^{+}$ and the reflexive-transitive closure of $\step{\tau}$ by $\step{}^{*}$.

\begin{definition}[Branching Bisimilarity]\label{def:bbisim}
Let $T_1=(\Sta_1,\step{}_1,\uparrow_1)$ and $T_2=(\Sta_2,\step{}_2,\uparrow_2)$ be transition systems. A \emph{branching bisimulation} from $T_1$ to $T_2$ is a binary relation $\R\subseteq\Sta_1\times\Sta_2$ such that for all states $s_1$ and $s_2$, $s_1\R s_2$ implies
\begin{enumerate}
    \item if $s_1\step{a}_1s_1'$, then there exist $s_2',s_2''\in\Sta_2$, s.t. $s_2\step{}_2^{*}s_2''\step{(a)}s_2'$, $s_1\R s_2''$ and $s_1'\R s_2'$;
    \item if $s_2\step{a}_2s_2'$, then there exist $s_1',s_1''\in\Sta_1$, s.t. $s_1\step{}_1^{*}s_1''\step{(a)}s_1'$, $s_1''\R s_2$ and $s_1'\R s_2'$.
\end{enumerate}
The transition systems $T_1$ and $T_2$ are \emph{branching bisimilar} (notation: $T_1\bbisim T_2$) if there exists a branching bisimulation $\R$ from $T_1$ to $T_2$ s.t. $\uparrow_1\R\uparrow_2$.

A branching bisimulation $\R$ from $T_1$ to $T_2$ is \emph{divergence-preserving} if, for all states $s_1$ and $s_2$, $s_1\R s_2$ implies
\begin{enumerate}
\setcounter{enumi}{2}
    \item if there exists an infinite sequence $(s_{1,i})_{i\in\mathbb{N}}$ s.t. $s_1=s_{1,0},\,s_{1,i}\step{\tau}s_{1,i+1}$ and $s_{1,i}\R s_2$ for all $i\in\mathbb{N}$, then there exists a state $s_2'$ s.t. $s_2\step{}^{+}s_2'$ and $s_{1,i}\R s_2'$ for some $i\in\mathbb{N}$; and
    \item if there exists an infinite sequence $(s_{2,i})_{i\in\mathbb{N}}$ s.t. $s_2=s_{2,0},\,s_{2,i}\step{\tau}s_{2,i+1}$ and $s_1\R s_{2,i}$ for all $i\in\mathbb{N}$, then there exists a state $s_1'$ s.t. $s_1\step{}^{+}s_1'$ and $s_1'\R s_{2,i}$ for some $i\in\mathbb{N}$.
\end{enumerate}
The transition systems $T_1$ and $T_2$ are \emph{divergence-preserving branching bisimilar} (notation: $T_1\bbisim^{\Delta}T_2$) if there exists a divergence-preserving branching bisimulation $\R$ from $T_1$ to $T_2$ s.t. $\uparrow_1\R\uparrow_2$.
\end{definition}

\myparagraph{Proof of Theorem~\ref{thm:itm-lts-iso}}

In the proof it is convenient to allow RTMs to have transitions of the form $s\step{a[d/e]S}t$, where $S$ is a stay transition with no movement of the tape head. We refer to such machines as RTMs with stay transitions. The operational semantics of RTMs can be extended to an operational semantics for RTMs with stay transitions by adding the clause: $(s,\delta_L\check{d}\delta_R)\step{a}(t,\delta_L \check{e}\delta_R)$ iff $s\step{a[d/e]S}t$. The transition system of an RTM with stay transitions can be simulated by an RTM up to divergence-preserving branching bisimilarity.

\begin{lemma}~\label{lemma:RTMN}
The transition system associated with an RTM with stay transitions is executable up to divergence-preserving branching bisimilarity.
\end{lemma}

\begin{proof}
We suppose that $\M=(\Sta,\step{},\uparrow)$ is an RTM with stay transitions, and its transition system is $\T(\M)$. We define a normal RTM $\M'=(\Sta_1,\step{}_1,\uparrow_1)$ that simulates $\T(\M)$ as follows:

\begin{enumerate}
\item $\Sta_1=\Sta\cup\{s_t\mid s,t\in\Sta\}$; 
\item $s\step{a[d/e]L}_1 t$ iff $s\step{a[d/e]L}t$;
\item $s\step{a[d/e]R}_1 t$ iff $s\step{a[d/e]R}t$;
\item $s\step{a[d/e]L}_1 s_t$ and $s_t\step{\tau[d/d]R}_1t$ iff $s\step{a[d/e]S}t$; and
\item $\uparrow_1=\uparrow$.
\end{enumerate}

Then it is straight forward to $\T(\M')\bbisimd\T(\M)$.
\end{proof}

We proceed to give a proof for Theorem~\ref{thm:itm-lts-iso}

\begin{proof}
    We let $\I=(Q,\step{}_{\I},\qin)$ be an ITM. By Lemma~\ref{lemma:RTMN}, it is enough to show that there exists an RTM with stay transitions $\M$ satisfying $\T(\M)\bbisimd\T(\I)$.
    We construct $\M=(\Sta,\step{},\uparrow)$ as follows:
    \begin{enumerate}
        \item $\Sta=\In\cup\Out$, where $\In=Q$ and $\Out=\{s_o\mid o\in\Sigma_{\lambda}, s\in Q\}$ as defined in Definition~\ref{def:lts-itm};
        \item the transition relation $\step{}$ is defined by:
        $s\step{\mathit{in}(i)[d/e]M}t_o$ if $(s,d,i)\step{}_{\I}(t,e,M,o)$, and $s_o\step{\mathit{out}(o)[e/e]S}s$ for all $s\in\Sta$, $o\in\Sigma_{\lambda}$; and
        \item $\uparrow=\qin$.
    \end{enumerate}
    Then according to Definitions~\ref{def:lts-tm} and~\ref{def:lts-itm}, we get a transition system $\T(\M)=\T(\I)$, where `=' is the pointwise equality, which also implies $\T(\M)\bbisimd \T(\I)$.
\end{proof}

\myparagraph{Proof of Lemma~\ref{lemma:i-o-lts}}
\begin{proof}
 A state in $\Sta_{\M}$ is a configuration $(s,\delta)$ of $\M$, and we can make a partition of the set of all configurations according to the control states. If $s\in\In$, then $(s,\delta)\in\In_{\M}$; if $s\in\Exe$, then $(s,\delta)\in\Exe_{\M}$, where $\In$ and $\Exe$ are defined in Definition~\ref{def:rtm-omega}.
\begin{enumerate}
    \item (Alternation) By condition 1 in Definition~\ref{def:rtm-omega}, we have $\Sta=\In\cup\Exe$ and $\In\cap\Exe=\emptyset$, which infers $\Sta_{\M}=\In_{\M}\cup\Exe_{\M}$ and $\In_{\M}\cap\Exe_{\M}=\emptyset$; moreover, by condition 2, for a transition $s\step{a[d/e]M}t$, if $s\in\In$, then $a\in \{?0,?1\}$ and $t\in\Exe$; if $s\in\Exe$, then $a\in\{!0,!1,\tau\}$ and $t\in \In$, which infers that for every transition $s\step{a} s'$, if $s\in\In_{\M}$, then $a\in\{?0,?1\}$ and $s'\in\Exe_{\M}$; if $s\in\Exe_{\M}$, then $a\in\{!0,!1,\tau\}$ and $s'\in\In_{\M}$.
    \item (Unambiguity) By condition 3 in Definition~\ref{def:rtm-omega}, for all $(s,d)$ where $s\in\Exe$ and $d\in\Dbox$, there is at most one transition $s\step{o[d/e]M}t$, which infers that for every $s\in\Exe_{\M}$, there is exactly one outgoing transition $s\step{a} s'$ with $a\in\{!0,!1,\tau\}$.
    \item (Totality) By condition 4 in Definition~\ref{def:rtm-omega}, for all $(s,d)$ where $s\in\In$ and $d\in\Dbox$, there are exactly two transitions of the form $s\step{i[d/e]M}t$, with $?0$ and $?1$ as there labels, respectively, which infers that for every $s\in\In_{\M}$, there are two outgoing transitions labelled by $?0$ and $?1$, respectively.
\end{enumerate}
\end{proof}

\myparagraph{Proof of Lemma~\ref{lemma:io-LTS-omega}}

\begin{proof}
   Let $\T$ be an i/o interactive transition system, and let $s_0$ be the initial state of $\T$. By Definition~\ref{def:in-out-lts}, we need to show that there exists an $\omega$-translation $\phi$ such that for every $\vec{x}\in\Sigma^{\omega}$, there exists a trace $\sigma\in\witr(s_0)$ with input stream $\vec{x}$, and for every trace with input stream $\vec{x}$, its output stream is $\vec{y}=\phi(\vec{x})$.

   By the alternation condition in Lemma~\ref{lemma:i-o-lts}, every $\sigma\in\witr(s_0)$ is of the form $i_0 o_0 i_1 o_1\ldots$ where $i_j\in\{?0,?1\}$ and $o_j\in\{!0,!1,\tau\}$.
   Let $\vec{x}$ be an arbitrary input stream, by the totality condition in Lemma~\ref{lemma:i-o-lts}, we can find a trace $\sigma\in\witr(s_0)$ with input stream $\vec{x}$.

    Moreover, given an trace $\sigma$ with an infinite input stream $\vec{x}$, by interactiveness, it would always produce an infinite output stream $\vec{y}$.

    Finally, by unambiguity, there do not exist two traces sharing the same input stream. It follows that for every trace with input stream $\vec{x}$, its output stream is $\vec{y}$. Hence, we relate with every input stream a unique output stream, in a way, we get a $\omega$-translation from $\T$.
\end{proof}

\myparagraph{Proof of Theorem~\ref{thm:lc-exe-translation2}}

It is not hard to show the following lemmas,

\begin{lemma}\label{lemma:bis-translation}
Let $\T_1$ and $\T_2$ be two i/o transition systems, and $\T_1\bbisim \T_2$. Then they realize the same $\omega$-translation.
\end{lemma}

\begin{proof}
We let $s_1$ and $s_2$ be the initial states of $\T_1$ and $\T_2$, respectively. As $\T_1\bbisim \T_2$, we have that for every $\sigma\in\witr(s_1)$, there exists a trace $\sigma'\in\witr(s_2)$, and they share the same input and output stream, and vice versa. It follows that $\T_1$ and $\T_2$ realize the same $\omega$-translation.
\end{proof}

\begin{lemma}~\label{lemma:finite-input}
Let $\T$ be an interactive i/o transition system, and let $s_0$ be its initial state, then the following function is computable:
$g: \Sigma^{*}\rightarrow \Sigma^{*}$, satisfying that if $g(x)=y$, then for every $\sigma\in\witr(s_0)$ with input and output stream $\vec{x}$ and $\vec{y}$, if $x\prec\vec{x}$, then $y\prec\vec{y}$.
\end{lemma}

\begin{proof}
We consider a finite trace from $s_0$, we can associate with such a trace its input and output sequences in a similar way as defined in Definition~\ref{def:in-out-lts}.
By Lemma~\ref{lemma:i-o-lts}, there is only one finite trace with $x$ as its input sequence, and its output sequence is $y$. By totality, it holds for every $x\in\Sigma^{*}$.
As the transition relation of i/o transition systems are computable, $g$ is also computable.
\end{proof}

Hence, we are able to make a proof of Theorem~\ref{thm:lc-exe-translation2}
\begin{proof}
We let $\phi$ be an $\omega$-translation.
\begin{enumerate}
\item For the ``only if'' part, we need to show that there exists a computable total function $g:\Sigma^{*}\rightarrow\Sigma^{*}$, such that $g$ is monotonic and for all strictly increasing chains $u_1\prec u_2\prec\ldots\prec u_t\prec\ldots$ with $u_t\in\Sigma^{*}$ ($t\geq 1$), one has $\phi(\lim_{t\rightarrow\infty}u_t)=\lim_{t\rightarrow\infty}g(u_t)$.

    We assume that $\phi$ is realized by an interactive i/o transition system $\T$, and we let $s_0$ be the initial state of $\T$. By Lemma~\ref{lemma:finite-input} the following function is computable: $g: \Sigma^{*}\rightarrow \Sigma^{*}$, satisfying that if $g(x)=y$, then for every $\sigma\in\witr(s_0)$ with input and output stream $\vec{x}$ and $\vec{y}$, if $x\prec\vec{x}$, then $y\prec\vec{y}$. By unambiguity and totality, $g$ is a monotonic and total computable function.

 Moreover, for a strictly increasing chain $u_1\prec u_2\prec\ldots\prec u_t\prec\ldots$ with $u_t\in\Sigma^{*}$ for $t\geq 1$, the computation of $\lim_{t\rightarrow\infty}g(u_t)$ is the execution of a trace $\sigma$ receiving the input stream $\lim_{t\rightarrow\infty}u_t$. Hence we have $\phi(\lim_{t\rightarrow\infty}u_t)=\lim_{t\rightarrow\infty}g(u_t)$.

    Thus, $g$ is the computable total function we need, and it follows that $\phi$ is a computable limit-continuous total function.

\item

    For the ``if'' part, we assume that $\phi$ is a total limit-continuous function, and design an RTM $\M$ to realize this translation.
    By Theorem~\ref{thm:int-trans}, $\phi$ is interactively computable by some ITM $\M'$. According to Definition~\ref{def:lts-itm} and Lemma~\ref{lemma:i-o-lts}, the transition system associated with $\M'$ is an i/o transition system, moreover, according to Corollary~\ref{cor:itm-lts-exe}, it is an executable i/o transition system. Therefore, we have shown that $\phi$ is an executable $\omega$-translation by Lemma~\ref{lemma:bis-translation}.
\end{enumerate}
\end{proof}

\myparagraph{Proof of Theorem~\ref{thm:bound-lts-rtma}}

\begin{proof}
We assume that $\T=(\Sta_{\T},\step{}_{\T},\uparrow_{\T})$ is an $\Atau$-labelled transition system. It has $n$ distinct action labels and its branching degree is bounded by $k$.
Then we encode $\Atau$ and $\Sta_{\T}$ as natural numbers. Let $\encode{a}$ and $\encode{s}$ be the encodings of an action and a state, and $\encode{x_1,x_2,\ldots,x_n}$ be the encoding of an $n$-tuple.

The advice process $A_{f}$ realizes the following function:
\begin{equation*}
f(\encode{s})=\encode{a_1,\ldots, a_m, s_1,\ldots,s_m}
\enskip,
\end{equation*}
where $(a_i,s_i)\in \{(a_1,s_1),\ldots,(a_m,s_m)\}$ iff $s\step{a_i}_{\T} s_i$.

An outline of the execution of $\M$ is defined as follows.
\begin{enumerate}
    \item We need the following control states: $\mathit{initial}$, $\mathit{advice}$, $\mathit{decode}$, $\mathit{next}_{\Atau^{\leq k}}$ ($\Atau^{\leq k}$ ranges over all $\Atau$ words with at most length $k$), $\mathit{choose}_i$ ($i\leq k$).
    \item The execution of $\M$ is as follows, its initial configuration is $(\mathit{initial},\Box)$.
    \begin{enumerate}
        \item In $\mathit{initial}$ state, the machine writes the encoding of initial state of the transition system $\encode{\uparrow_{\T}}$ on the tape, and reaches $\mathit{advice}$ state.
        \begin{equation*}
         (\mathit{initial}, \Box)\step{}^{*} (\mathit{advice}, \encode{\uparrow_{\T}})
         \enskip.
        \end{equation*}
        \item In $\mathit{advice}$ state, the machine sends the encoding of the current state $\encode{s_0}$ to the advice process, and gets the encoding of list of all possible transitions $\encode{a_1,\ldots, a_m, s_1,\ldots,s_m}$ from the advice process.
        \begin{equation*}
         (\mathit{advice},\encode{s_0})\step{}^{*} (\mathit{decode},\encode{a_1,\ldots, a_m, s_1,\ldots,s_m})
         \enskip.
        \end{equation*}
        \item In $\mathit{decode}$ state, the machine decodes all the actions from the tape, and enters one of the $\mathit{next}$ state.
        \begin{equation*}
        (\mathit{decode},\encode{a_1,\ldots, a_m, s_1,\ldots,s_m})\step{}^{*} (\mathit{next}_{\{a_1,\ldots,a_m\}},\encode{s_1,\ldots,s_m})
        \enskip.
        \end{equation*}
        \item In $\mathit{next}_{\{a_1,\ldots,a_m\}}$ state, the machine chooses one of the actions. For every $i=1,\ldots,m$, there is a transition
        \begin{equation*}
        (\mathit{next}_{\{a_1,\ldots,a_m\}},\encode{s_1,\ldots,s_m})\step{a_i} (\mathit{choose}_i,\encode{s_1,\ldots,s_m})
        \enskip.
        \end{equation*}
        \item In $\mathit{choose}_i$ state, the machine projects the encoding $\encode{s_1,\ldots,s_m}$ to the encoding of the $i$-th state, and enters $\mathit{advice}$ state again.
        \begin{equation*}
        (\mathit{choose}_i,\encode{s_1,\ldots,s_m})\step{}^{*}(\mathit{advice},\encode{s_i})
        \enskip.
        \end{equation*}
    \end{enumerate}
\end{enumerate}
The above procedure describes the simulation of a step of transition $s_0\step{a_i}_{\T}s_i$ in $\T$. Note that the choice of the transition is happened only in the state $\mathit{next}_{\{a_1,\ldots,a_m\}}$. Moreover, no infinite $\tau$-transition sequence is introduced for simulation.
Hence, we are able to verify that $\T([\M\parallel A_{f}]_{\C})\bbisimd \T$.
\end{proof}

\myparagraph{Proof of Theorem~\ref{thm:lts-rtma}}

\begin{proof}
We assume that $\T=(\Sta_{\T},\step{}_{\T},\uparrow_{\T})$ is a countable $\Atau$-labelled transition system. It has $n$ distinct action labels and it possibly has infinitely branching.
Then we encode $\Atau$ and $\Sta_{\T}$ as natural numbers. Let $\encode{a}$ and $\encode{s}$ be the encodings of an action and a state, and $\encode{x_1,x_2,\ldots,x_n}$ be the encoding of an $n$-tuple.

The transition relation $\step{}_{\T}$ maps a state, namely, $s_0$, to a possibly infinite set $\{(a_i,s_i)\mid s_0\step{a_i}_{\T} s_i\}$, denoted by $s_0\step{}_{\T}$. We define an order $<_{\T}$ over the elements in the set $s_0\step{}_{\T}$ such that $(a,s)<_{\T}(a',s')$, if $\encode{a,s}<_{\T}\encode{a',s'}$.

The advice function $A_{f}$ realizes the following function:
\begin{equation*}
f(\encode{s_0,i})=\encode{a_i,s_i}
\enskip,
\end{equation*}
where $(a_i,s_i)$ is the $i$-th element from $s_0\step{}_{\T}$ regarding to $<_{\T}$.

An outline of the execution of $\M$ is defined as follows.
\begin{enumerate}
    \item We need the following control states: $\mathit{initial}$, $\mathit{advice}$, $\mathit{decode}$, $\mathit{next}_{\Atau}$, $\mathit{choose}_i$ ($i=1,2$).
    \item The execution of $\M$ is as follows, we use a pair $(s,\delta)$ to denote the current configuration of the machine.
    \begin{enumerate}
        \item In $\mathit{initial}$ state, the machine writes the encoding of the initial state of the transition system $\encode{\uparrow_{\T}}$ on the tape, and reaches $\mathit{advice}$ state.
        \begin{equation*}
         (\mathit{initial}, \Box)\step{}^{*} (\mathit{advice}, \encode{\uparrow_{\T},1})
         \enskip.
        \end{equation*}
        \item In $\mathit{advice}$ state, the machine either increase the counter $i$ by $1$, or sends $\encode{s_0,i}$ to the advice, and gets $\encode{(a_i,s_i)}$ from the advice.
        \begin{eqnarray*}
         (\mathit{advice},\encode{s_0,i})\step{}^{*}(\mathit{advice},\encode{s_0,i+1}), or\\
         (\mathit{advice},\encode{s_0,i})\step{}^{*} (\mathit{decode},\encode{s_0,s_i,a_i})
         \end{eqnarray*}
        \item In $\mathit{decode}$ state, the machine decodes the action $a_i$ from the tape, and enters the state $\mathit{next_{a_i}}$.
        \begin{equation*}
        (\mathit{decode},\encode{s_0,s_i,a_i})\step{}^{*} (\mathit{next}_{a_i},\encode{s_0,s_i})
        \enskip.
        \end{equation*}
        \item In $\mathit{next}_{a_i}$ state, the machine either performs the action, or change its current choice to another transition.
        \begin{eqnarray*}
        (\mathit{next}_{a_i},\encode{s_0,s_i})\step{\tau} (\mathit{choose}_1,\encode{s_0,s_i}), or\\
        (\mathit{next}_{a_i},\encode{s_0,s_i})\step{a_i} (\mathit{choose}_2,\encode{s_0,s_i})
        \enskip.
        \end{eqnarray*}
        \item In $\mathit{choose}_i$ state (i=1,2), the machine projects the encoding $\encode{s_1,s_2}$ to the encoding of the $i$-th state, and enters $\mathit{advice}$ state again.
        \begin{equation*}
        (\mathit{choose}_i,\encode{s_1,s_2})\step{}^{*}(\mathit{advice},\encode{s_i,1})
        \enskip.
        \end{equation*}
    \end{enumerate}
\end{enumerate}
One can verify that $\R=\{(s,s')\mid s\in\Sta_{\T}, s'=\\(\mathit{advice}, \encode{s,i})\mbox{ or } (\mathit{decode},\encode{s,a_i,s_i}) \mbox{ or } (\mathit{next}_{a_i},\encode{s,s_i})\mbox{ or } (\mathit{choose}_1,\encode{s,s_i})\mbox{ or } (\mathit{choose}_2,\encode{s_i,s})\}$ is a branching bisimulation relation. Hence, we have $\T([\M\parallel A_{f}]_{\C})\bbisim T$.
\end{proof} 